\pdfoutput=1
\documentclass[sigconf,nonacm]{acmart}

\usepackage{amsthm}
\usepackage{tikz}
\usetikzlibrary{patterns}
\usepackage{subcaption}
\usepackage{cleveref}
\usepackage{algorithm}
\usepackage{algpseudocode}
\usepackage{float}

\newcommand{\gapLo}[1]{\MinDist^2(\check{O}_{#1}, B_{#1}) - \MaxDist^2(\check{O}_{#1}, E_{#1})}
\newcommand{\gapHi}[1]{\MinDist^2(\hat{O}_{#1}, B_{#1}) - \MaxDist^2(\hat{O}_{#1}, E_{#1})}

\DeclareMathOperator{\Corners}{Corners}
\DeclareMathOperator{\MaxDist}{MaxDist}
\DeclareMathOperator{\MinDist}{MinDist}
\DeclareMathOperator{\dist}{dist}
\DeclareMathOperator*{\argmin}{\arg\!\min}
\DeclareMathOperator*{\argmax}{\arg\!\max}

\begin{document}

\title{Optimal Bounds-Only Pruning for Spatial AkNN Joins}

\author{Dominik Winecki}
\email{winecki.1@osu.edu}
\orcid{0009-0009-8632-3102}
\affiliation{
  \institution{The Ohio State University}
  \city{Columbus}
  \state{Ohio}
  \country{USA}
}

\begin{abstract}
We propose a bounds-only pruning test for exact Euclidean AkNN joins on partitioned spatial datasets.
Data warehouses commonly partition large tables and store row group statistics for them to accelerate searches and joins, rather than maintaining indexes.
AkNN joins can benefit from such statistics by constructing bounds and localizing join evaluations to a few partitions before loading them to build spatial indexes.
Existing pruning methods are overly conservative for bounds-only spatial data because they do not fully capture its directional semantics, thereby missing opportunities to skip unneeded partitions at the earliest stages of a join.
We propose a three-bound proximity test to determine whether \textit{all points} within a partition have a closer neighbor in one partition than in another, potentially occluded partition.
We show that our algorithm is both optimal and efficient.
\end{abstract}


\maketitle

\section{Introduction}

K-Nearest-Neighbor (kNN) searches are ubiquitous operations in spatial data analysis.
Given their significance and simplicity, K-Nearest-Neighbor and All-K-Nearest-Neighbor (AkNN) joins have been thoroughly studied.
Most of this work focuses on building and searching spatial indexes; executing AkNN joins on unindexed or minimally-indexed datasets has received less study.

We focus on executing AkNN joins on unindexed partitioned spatial datasets.
For example, in the data warehouse architecture, large tables are split into many partitioned files in an object store.
Such a pattern has also gained traction at a smaller scale, such as using DuckDB to query a directory of Parquet files on a laptop.
This style of data analysis has gained popularity for the performance benefits of avoiding index maintenance on inserts and updates~\cite{10.1145/3226595.3226638} and for pragmatic reasons, such as the portability and interoperability of the storage layer and the ability to download and query datasets without preprocessing.
Indexes are still used, but are created only temporarily for use during a running query.
Some metadata is available in the serialization layer to accelerate such queries via row-group statistics.
For example, Parquet stores minimum and maximum values for each row group, allowing joins to skip scans on row groups or partitions by pushing down filter operations.
We focus on AkNN queries across spatial datasets in these partitioned data storage architectures, where aggregation statistics metadata are available, and there is a high cost for missing an opportunity to prune: reading, decompressing, and inserting thousands of rows into an index.

Partition pruning can be used as the first stage of an AkNN join, pruning partitions before reading any rows or constructing indexes.
Then, a subsequent algorithm can be used to execute an AkNN on the underlying partitions.
Existing kNN pruning methods are suboptimal for this first stage; they fail to capture directionality by constructing a single distance range from two bounds, resulting in unnecessarily pessimistic pruning.
Informally, when there is a partition ``in between'' the origin partition and a partition further out, the middle partition allows pruning of the further out partition.

We propose a simple three-bound test to determine if, for partitions $O$, $E$, and $B$, all points in $O$ are closer to all points in $E$ than to any point in $B$.
We prove this is the optimal way to prune an AkNN join when only partition bounds are considered.
Our test either asserts that a partition can be pruned or constructs example points demonstrating why a partition bound cannot be pruned.
This early-stage bounds-only pruning complements and is orthogonal to subsequent per-point AkNN, including approximate methods, since approximate methods cannot make assumptions about the distribution of unindexed data.
This makes our approach a simple, optimal, and pragmatic improvement for modern spatial data warehouses.

In this paper, we propose a three-bound test algorithm and prove its correctness and optimality.
We then provide an optimized version of our algorithm that reduces the runtime to $O(R)$, where $R$ is the number of dimensions, matching the information-theoretic lower bound.
Following this, we show that our bound proximity test constitutes a strict partial order, allowing us to determine the optimal order for loading partitions in the general case.

\section{Related Work}

\subsection{Spatial Data Warehouses}
OLAP database trends have prioritized data layout and compression over extensive use of indexes~\cite{10.1145/3226595.3226638,ziauddin2017dimensions}.
Increasingly higher-core-count CPUs and GPUs are becoming available, as are CPUs with SIMD instructions, prioritizing brute-force parallel processing over index traversal.
kNN joins on GPUs are faster on simple data layouts~\cite{velentzas_-memory_2020,attiogbe_gpu-based_2021,amdal_top-k_2020,gowanlock_hybrid_2021}.
Many relational OLAP systems using columnar stores maintain distribution information for groups of rows, notably Snowflake~\cite{dageville_snowflake_2016} and AWS Redshift.
These are sometimes referred to as zone maps~\cite{ziauddin2017dimensions}.
Similarly, many spatial data warehouses implement AkNN using a MapReduce~\cite{10.1145/1327452.1327492} execution architecture~\cite{lu_efficient_2012,zhang_efficient_2012}, often in Spark~\cite{maillo_knn-is_2017,shangguan_big_2017,tang_locationspark_2016,vu_incremental_2021,yu_geospark_2015}.
These systems use partition-aware data processing to efficiently query data in a shared disk/object storage layer.

\subsection{Nearest Neighbors}
Spatial nearest-neighbor, k-nearest-neighbor, and all-k-nearest-neighbor queries have been extensively studied for spatial use cases~\cite{eldawy_era_2016,pandey_how_2018}.
We constrain our focus to spatial low-dimensional exact AkNN with Euclidean/L2 distance.
Partition-bound pruning is the earliest stage of executing an AkNN join and is orthogonal to the per-point execution at a later stage, which may go beyond these constraints (e.g., approximate AkNN) but does not impact the earlier partition pruning.
Additionally, most kNN research focuses on creating spatially disjoint partitions~\cite{baiyou_qiao_boundary_2015,chatzimilioudis_distributed_2016}, such as quadtrees~\cite{tang_locationspark_2016}, but our use case often has highly overlapping partitions.
Similarly, most kNN pruning research focuses on point-level pruning against an index; however, in this application, we are strictly interested in pruning partitions before loading any points or constructing indexes.

\subsubsection{$k=1$ Point-to-Bound}

\begin{figure}[H]
    \centering
    \begin{tikzpicture}[>=stealth, line cap=round, line join=round]

  \fill (-1,-0.3) circle (2pt);

  \draw[thick] (2,-1) rectangle (3,1.5);

  \draw[<->, dashed] 
    (-1,-0.3) -- node[above]{MinDist} (2,-0.3);

  \draw[<->, dashed] 
    (-1,-0.3) -- node[sloped,above]{MaxDist} (3,1.5);

  \draw[<->, dashed] 
    (-1,-0.3) -- node[sloped,below]{MinMaxDist} (3,-1);

\end{tikzpicture}
    \caption{Point-to-Bound Distance Measures. Minimum and maximum distance to a partition shown, as well as the special case maximum distance for $k=1$ on an AABB (MinMaxDist).}
    \label{fig:maxdist_mindist}
\end{figure}
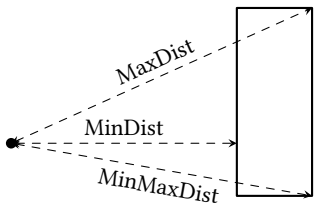

Seminal work on pruning nearest-neighbor queries by Roussopoulos et al. introduced querying via the $\MinDist$ and $\MaxDist$ functions~\cite{roussopoulos_nearest_1995}.
These functions both return distance bounds from a point to an AABB.
The MinMaxDist function can be used to find a tighter bound, but only if the bound is an MBR and $k=1$ (which also prevents usage when any predicates are needed).
Many efficient AkNN algorithms build upon Point-to-Bound pruning~\cite{xia_gorder_2004,xie_simba_2016}.

\subsubsection{Bound-to-Bound}

\begin{figure}[H]
    \centering
    \begin{tikzpicture}[>=stealth, line cap=round, line join=round]
      \draw[thick] (0,0) rectangle (2,2);
      \draw[thick] (4,1) rectangle (5,2);
      \draw[<->, dashed] (2,1.75) -- node[above]{BMinDist} (4,1.75);
      \draw[<->, dashed] (0,0) -- node[sloped,below]{BMaxDist} (5,2);
      \draw[<->, dashed] (0,0) -- node[sloped,above]{NXNDist} (4,2);
    \end{tikzpicture}
    \caption{Bound-to-Bound Distance Measures}
    \label{fig:pMinDist_pMaxDist}
\end{figure}
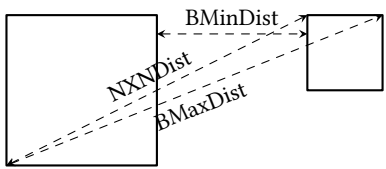

Methods that consider pruning strictly between bounds emerged later.
These were used by Sankaranarayanan et al.~\cite{sankaranarayanan_fast_2007} to define bound-to-bound versions of the min and max dist functions, henceforth referred to as BMinDist/BMaxDist, shown in \Cref{fig:pMinDist_pMaxDist}.
One could sort by BMaxDist and select partitions in this order until $k$ is saturated.
Then, only partitions with a BMinDist less than or equal to the saturating partition's BMaxDist (their PruneDist) must be searched.
They demonstrated that this is optimal for a two-partition test and that it has efficient implementations when using an index.

Additionally, the NXNDist~\cite{4221754} (short for MinMaxMinDist) is the bound-to-bound equivalent of MinMaxDist\footnote{MinMaxDist refers to the Roussopoulos et al. definition~\cite{roussopoulos_nearest_1995}, which differs from a distance metric by the same name in the NXNDist paper~\cite{4221754}, which, as they note, refers to the closest pairs definition~\cite{10.1145/342009.335414}, not a nearest neighbors definition. Similarly, their MinMinDist is our BMinDist and their MaxMaxDist is our BMaxDist.}.
Similar to its point-to-bound counterpart, NXNDist only works for $k=1$, with minimum bounds and no additional predicates.

\subsubsection{MBR/AABB Terminology}
Most kNN research uses the term Minimum Bounding Rectangles (MBRs).
However, Axis-Aligned Bounding Boxes (AABBs) are the more appropriate term for this paper since we require bounding boxes to be axis-aligned but do not require minimum bounds.
In most prior work, MBRs are actually AABBs; only the MinMaxDist~\cite{roussopoulos_nearest_1995} and NXNDist~\cite{4221754} metrics require minimum bounds.
Notably, omitting this minimum bounds requirement allows for pruning on datasets partitioned by a spatial hash, such as H3~\cite{noauthor_h3_2018}.
In the absence of row group distribution statistics, an AABB can be derived from a spatial hash, but not an MBR.

\section{Problem Statement}

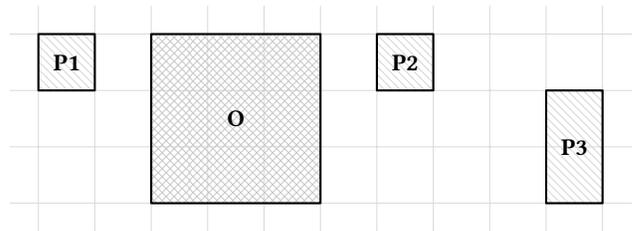
\begin{figure}
    \centering
    \begin{tikzpicture}[scale=0.75,>=stealth, line cap=round, line join=round]
        \begin{scope}
          \clip (-5.5, -0.5) rectangle (5.5, 3.5);
    
          \draw[
            step=1,
            lightgray,
            help lines,
            opacity=0.3
          ] (-6,-6) grid (6,6);
        \end{scope}
    
      \fill[pattern=crosshatch, pattern color=black!20]
           (-3,0) rectangle (0,3);
    
      \fill[pattern=north west lines, pattern color=black!15]
           (-5,2) rectangle (-4,3);      
      \fill[pattern=north west lines, pattern color=black!15]
           ( 1,2) rectangle ( 2,3);      
      \fill[pattern=north west lines, pattern color=black!15]
           ( 4,0) rectangle ( 5,2);      
    
        \draw[thick] (-5,2) rectangle (-4,3);  
        \draw[thick] (-3,0) rectangle (0,3);   
        \draw[thick] (1,2) rectangle (2,3);    
        \draw[thick] (4,0) rectangle (5,2);    
    
        \node at (-4.5, 2.5) {\textbf{P1}};
        \node at (-1.5, 1.5) {\textbf{O}};
        \node at ( 1.5, 2.5) {\textbf{P2}};
        \node at ( 4.5, 1.0) {\textbf{P3}};
    \end{tikzpicture}
    \caption{Spatial partition layout with three neighbor search partitions P1-P3. One partition, \textbf{P3}, can be pruned \textit{without visitation} because another partition, \textbf{P2}, will have all points closer to O. Only our proposed three-partition test identifies this.}
    \label{fig:direction_contradiction}
\end{figure}

Our task is to determine if all points in one bound are closer than all points in another bound within an AkNN join originating from a third bound:
$$\forall o \in O : \forall e \in E : \forall b \in B : \dist(o, e) < \dist(o, b)$$

Consider the partition layout in \Cref{fig:direction_contradiction}.
One dataset has a single partition, \textbf{O}; the neighbor candidate dataset has three partitions: \textbf{P1}, \textbf{P2}, and \textbf{P3}.
Which of these three partitions may contain the nearest neighbors to \textbf{O}?
Clearly, \textbf{P1} and \textbf{P2} must be considered, but \textbf{P3} is not a trivial case.
Intuitively, partition \textbf{P2} is roughly in between \textbf{O} and \textbf{P3}, which \textit{may} preclude the need to load \textbf{P3}.
The Bound-to-Bound test does not allow for removing \textbf{P3}.
In the bound-to-bound method, \textbf{P1} and \textbf{P2} are equidistant from \textbf{O} with no regard for \textbf{P2} being in the same direction as \textbf{P3}.
In the case of \Cref{fig:direction_contradiction}, this intuition is correct; \textbf{P3} does not need to be read. However, no existing spatial algorithms can show this prior to loading data from these partitions.

\section{AllPointsCloser Partition Pruning}\label{sec:pruning}

We are interested in identifying if one AABB partition $O$ (for origin) will always have a nearest neighbor in partition $E$ (for evaluation) instead of another partition $B$ (for basis):

$$\forall o \in O : \forall e \in E : \forall b \in B : \dist(o, e) < \dist(o, b)$$

We propose an equivalence in the \textit{All-Points Proximity Theorem}:
\begin{gather*}
    \left[\forall o' \in \Corners(O) : \MaxDist(o', E) < \MinDist(o', B)\right] \iff \\
    \forall o \in O : \forall e \in E : \forall b \in B : \dist(o, e) < \dist(o, b)
\end{gather*}

\begin{algorithm}[tbp]
    \caption{\textsc{AllPointsCloser}}
    \label{fig:algo}
    \begin{algorithmic}[1]
        \Statex \textbf{Input:} origin AABB $O$, evaluation AABB $E$, basis AABB $B$
        \Statex \textbf{Output:} \textbf{true} if all points in $O$ have nearest neighbor in $E$
        \Function{AllPointsCloser}{$O,E,B$}
            \ForAll{$p \in \Call{Corners}{O}$}
                \If{$\Call{MinDist}{p,B} \le \Call{MaxDist}{p,E}$}
                    \State \Return \textbf{false}
                \EndIf
            \EndFor
            \State \Return \textbf{true}
        \EndFunction
    \end{algorithmic}
\end{algorithm}

\Cref{fig:algo} directly performs this test.
Its correctness and optimality are evident from the previous if-and-only-if equivalence.
\Cref{fig:pruning_circles} shows a two-dimensional geometric interpretation of this distance test.
There are four circles, one centered at each corner, and each circle is only large enough to make contact with the nearest point in the basis partition.
The intersection of these circles forms a region which, from the perspective of the origin, is strictly closer than any points in the basis bound.
If at least $k$ points exist within this region, then the basis partition need not be searched.

\begin{figure}
    \centering
    \begin{tikzpicture}[scale=0.5,>=stealth,line cap=round,join=round]
        \draw[line width=1pt] (0,0) rectangle (4,4);

        \fill (0,0) circle (3pt);
        \fill (0,4) circle (3pt);
        \fill (4,0) circle (3pt);
        \fill (4,4) circle (3pt);
        \node[left]  at (0,0) {$O_1$};
        \node[right] at (4,0) {$O_2$};
        \node[right] at (4,4) {$O_3$};
        \node[left]  at (0,4) {$O_4$};

        \node at (2,2) {\large \textbf{O}};

        \draw[line width=1pt] (7,1) rectangle (8,2);

        \node at (7.5,1.5) {\large \textbf{B}};

        \draw (0,0) circle ({5*sqrt(2)});
        \draw (4,0) circle ({sqrt(10)});
        \draw (4,4) circle ({sqrt(13)});
        \draw (0,4) circle ({sqrt(53)});
        \begin{scope}
          \clip (0,0) circle ({5*sqrt(2)});
          \clip (4,0) circle ({sqrt(10)});
          \clip (4,4) circle ({sqrt(13)});
          \clip (0,4) circle ({sqrt(53)});
          \fill[pattern=north east lines] (-10,-10) rectangle (20,20);
        \end{scope}

        \draw[<->, dashed]
            (0,0) -- 
            node[midway, below, sloped] {\(\mathrm{MinDist}(O_1,B)\)}
            ++(225:{5*sqrt(2)});
        \draw[<->, dashed]
            (4,0) -- 
            node[midway, below, sloped] {\(\mathrm{MinDist}(O_2,B)\)}
            ++(315:{sqrt(10)});
        \draw[<->, dashed]
            (4,4) -- 
            node[midway, above, sloped] {\(\mathrm{MinDist}(O_3,B)\)}
            ++(45:{sqrt(13)});
        \draw[<->, dashed]
            (0,4) -- 
            node[midway, above, sloped] {\(\mathrm{MinDist}(O_4,B)\)}
            ++(135:{sqrt(53)});
    \end{tikzpicture}
    \vspace{-2mm}
    \caption{Geometric View of Pruning Test. Circles are drawn around each corner point of the origin partition to the nearest point in the basis partition; any geometry fully within their intersection must be closer than the basis partition.}
    \label{fig:pruning_circles}
\end{figure}
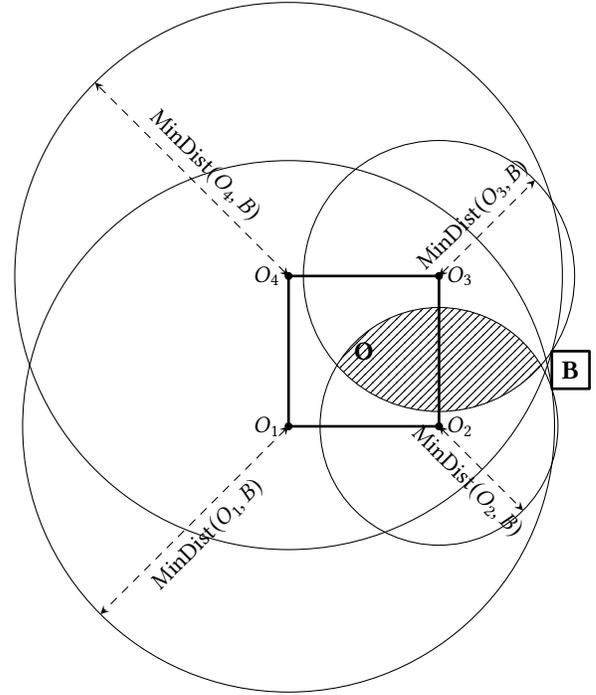

\subsection{All-Points Proximity Theorem}\label{sec:theorem_proof}

\subsubsection{Summary of Proof}
AABBs are convex sets in Euclidean space, and their corners represent the convex hull of their member points.
Therefore, we can represent any member point as an affine combination of the corners of an encompassing AABB.
We use the finite form of Jensen's inequality to constrain a convex function applied to all member points between its application on the corner points in \Cref{app_right}.
Specifically, we use the curve `MaxDist to $E$ subtracted by the MinDist to $B$', shown in \Cref{fig:img2}, which we prove is a convex function in \Cref{g_convex}.
Applying this convex difference-of-distances function to the convex hull is equivalent to evaluating all AABB member points.

\begin{figure}  
  \begin{subfigure}[b]{0.485\linewidth}
    \includegraphics[width=\linewidth]{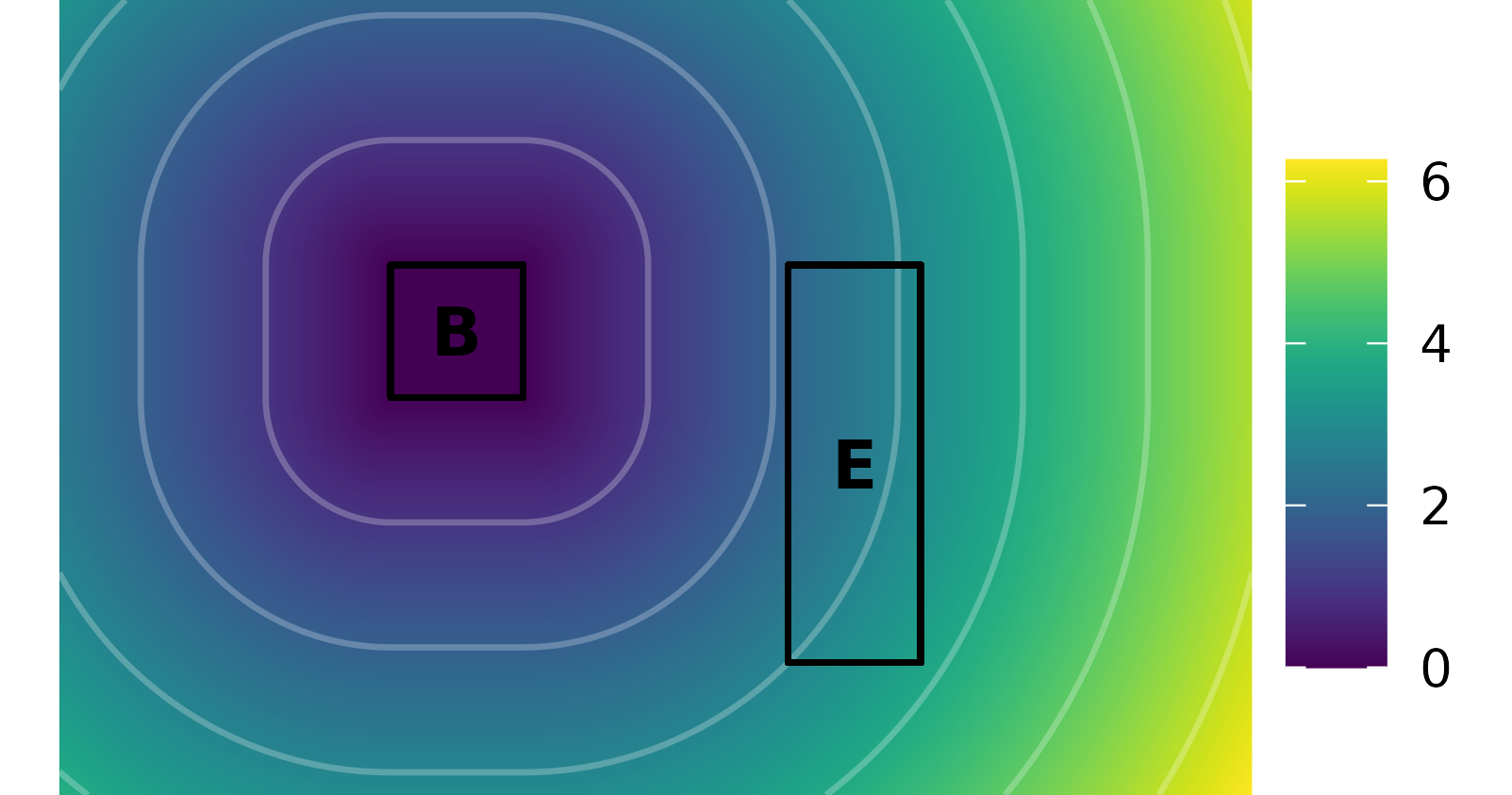}
    \caption{$\operatorname{MinDist}(p,B)$}
    \label{fig:img1}
  \end{subfigure}\hspace{0.02\linewidth}%
  \begin{subfigure}[b]{0.485\linewidth}
    \includegraphics[width=\linewidth]{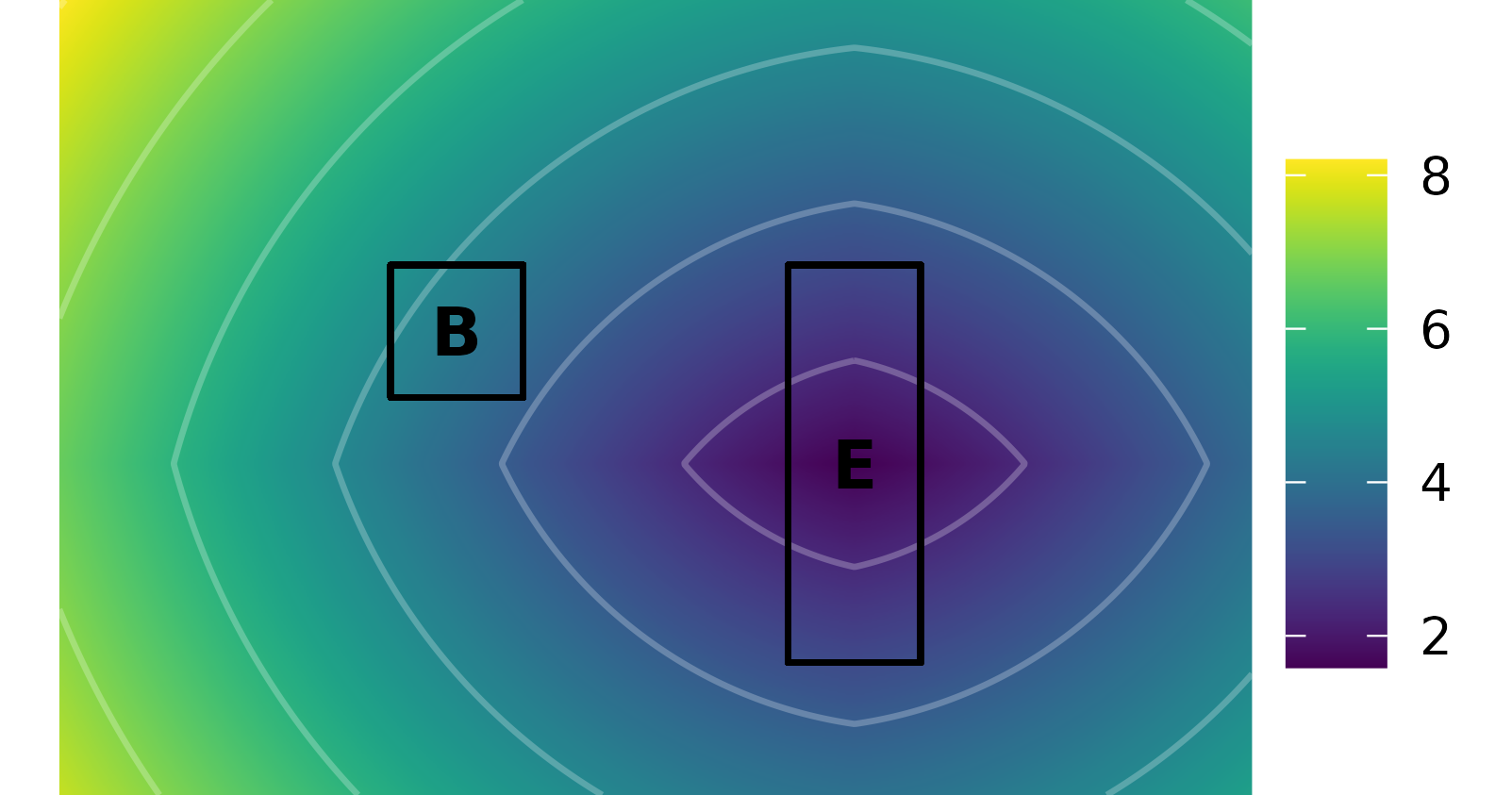}
    \caption{$\operatorname{MaxDist}(p,E)$}
    \label{fig:img3}
  \end{subfigure}

  \vspace{3mm}

  \begin{subfigure}[b]{0.485\linewidth}
    \includegraphics[width=\linewidth]{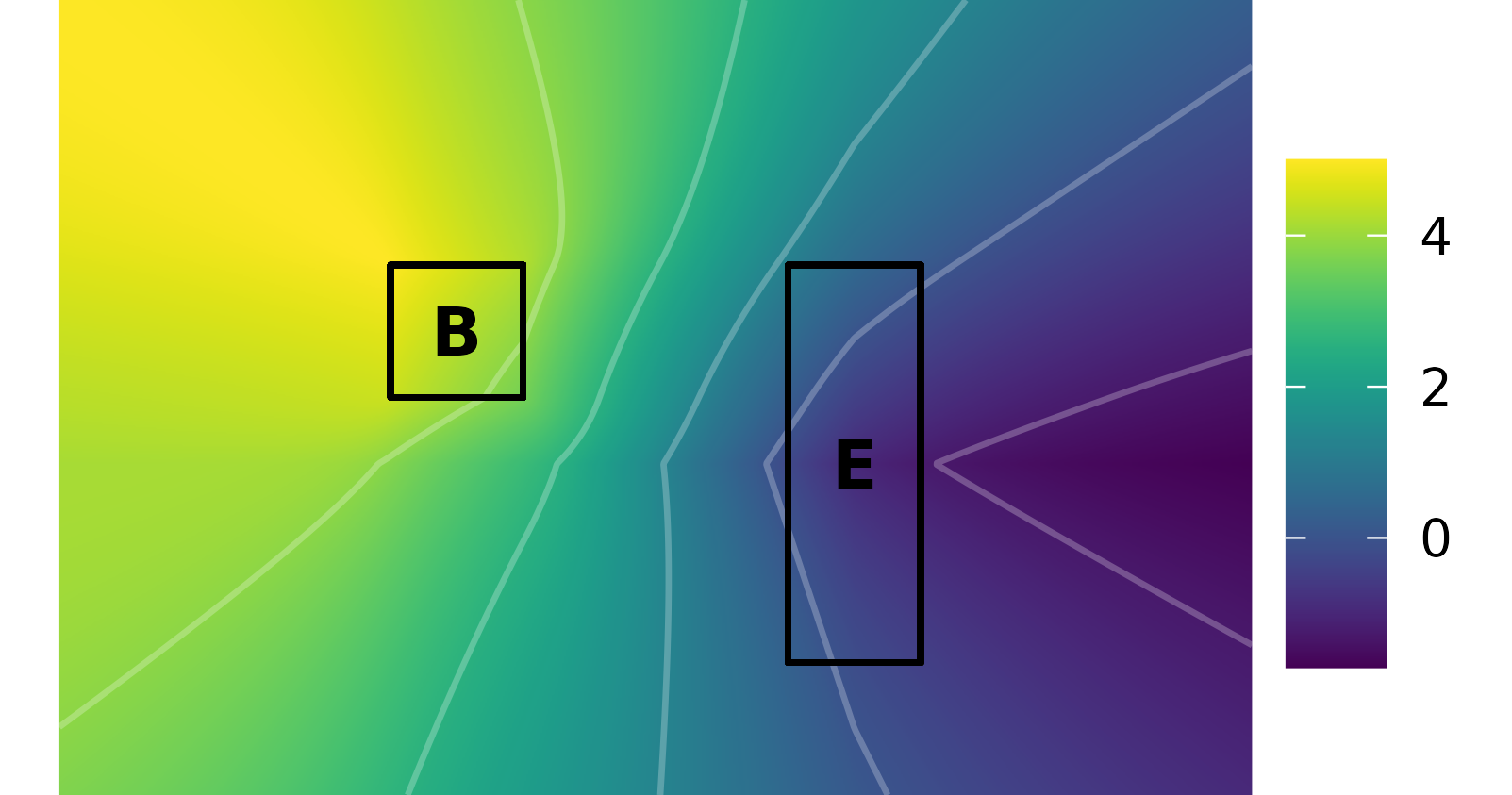}
    \caption{$\operatorname{MaxDist}(p,E)-\operatorname{MinDist}(p,B)$}
    \label{fig:img2}
  \end{subfigure}\hspace{0.02\linewidth}%
  \begin{subfigure}[b]{0.485\linewidth}
    \includegraphics[width=\linewidth]{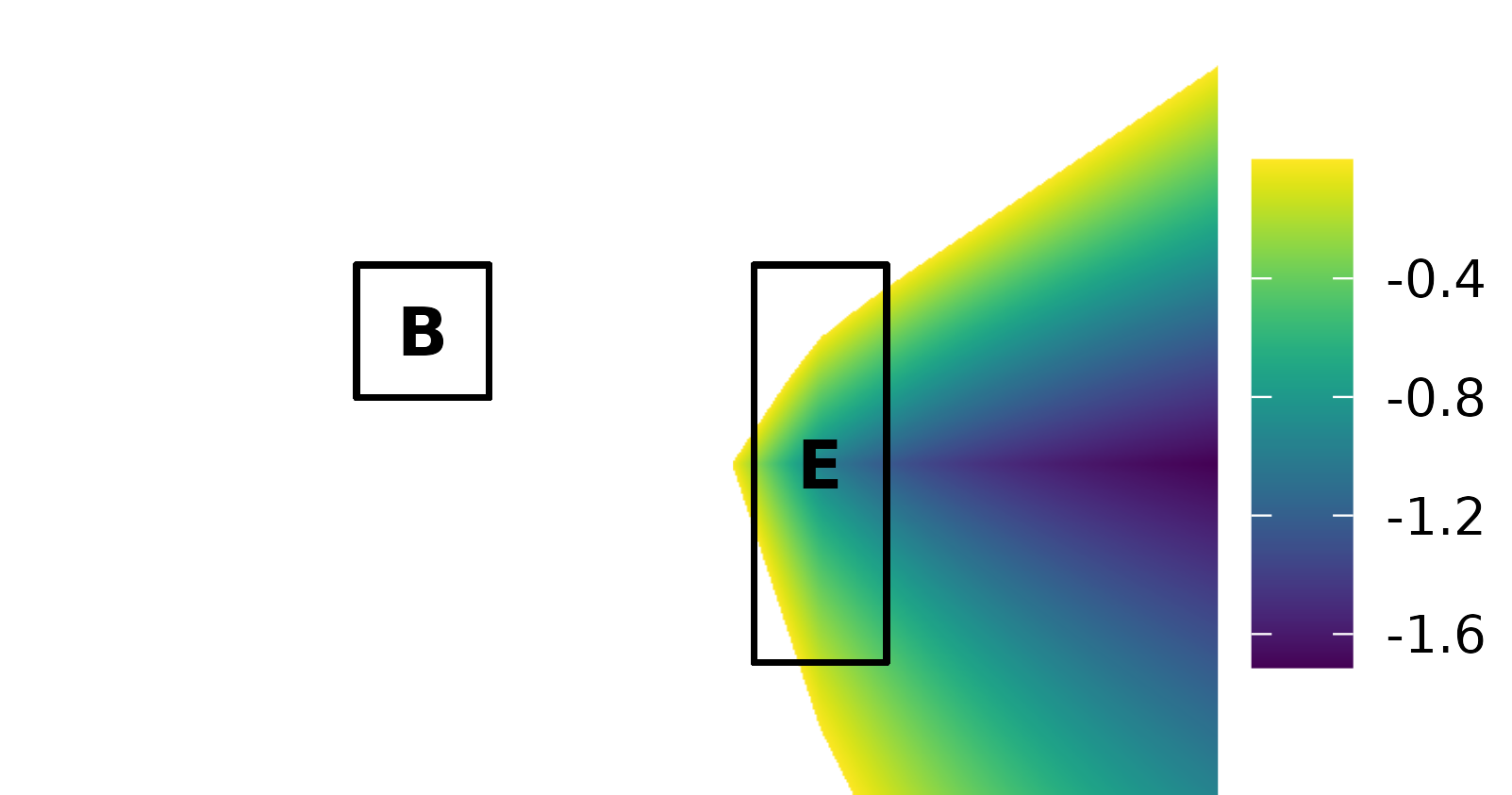}
    \caption{$\operatorname{MaxDist}(p,E)<\operatorname{MinDist}(p,B)$}
    \label{fig:img4}
  \end{subfigure}

  \caption{Visualization of used distance functions. \Cref{fig:img1} and \Cref{fig:img3} show existing point-to-bound functions. \Cref{fig:img2} shows our difference-of-distances function internally used by AllPointsCloser.}
  \label{fig:four-panel}
\end{figure}

\subsubsection{Definitions}
We have space $\mathbb{R}^r$.
This space contains origin $O$, basis $B$, and evaluation $E$ bounds, each an axis-aligned bounding box.
For an AABB $M$, $\Corners(M)$ represents the bound's $2^r$ corners/vertices.
$p \in M$ states that a point $p$ is within the inclusive bounds of $M$.

From prior work~\cite{roussopoulos_nearest_1995}, for a given AABB $M$,
\begin{equation}
\MaxDist(p, M) = \max_{q \in M} \dist(p, q)
\end{equation}
\begin{equation}
\MinDist(p, M) = \min_{q \in M} \dist(p, q)
\end{equation}
Additionally, $M_d$ is the $d$-th dimension bound of the AABB.
A dimension bound has low ($\check{M}_d$) and high ($\hat{M}_d$) bounds in its respective dimension.
All AABBs are well formed, so for each bound $\check{M}_d \leq \hat{M}_d$.

\subsubsection{Proof of All-Points Proximity Theorem}\footnote{A machine-checked proof in Lean 4 using Mathlib is available in the supplementary materials.}

\begin{theorem}[All-Points Proximity Theorem]
    \label{app}
    \begin{gather*}
    \left[\forall o' \in \Corners(O) : \MaxDist(o', E) < \MinDist(o', B)\right] \iff \\
    \forall o \in O : \forall e \in E : \forall b \in B : \dist(o, e) < \dist(o, b)
    \end{gather*}
\end{theorem}

\begin{proof}
    We prove each implication separately, with the straightforward reverse implication in \Cref{app_left} and the non-trivial forward implication in \Cref{app_right}.
\end{proof}

\begin{lemma}
    \label{app_left}
    \begin{gather*}
    \left[\forall o \in O : \forall e \in E : \forall b \in B : \dist(o, e) < \dist(o, b)\right] \implies \\ \forall o' \in \Corners(O) : \MaxDist(o', E) < \MinDist(o', B)
    \end{gather*}
\end{lemma}

\begin{proof}
    We assume $\forall o \in O : \forall e \in E : \forall b \in B : \dist(o, e) < \dist(o, b)$ is true and prove $\forall o' \in \Corners(O) : \MaxDist(o', E) < \MinDist(o', B)$.
    
    All points are inside their respective bounds:

\begin{equation}
o' \in \Corners(O) \implies o' \in O
\end{equation}
\begin{equation}
\MaxDist(o', E) = \max_{e\in E} \dist(o', e)
\end{equation}
\begin{equation}
\MinDist(o', B) = \min_{b\in B} \dist(o', b)
\end{equation}

    Let $o := o'$, $e := \argmax_{e\in E} \dist(o', e)$, and $b := \argmin_{b\in B} \dist(o', b)$.
    From their definitions, we know $o\in O$, $e\in E$, and $b\in B$.
    Rewriting with these values,

\begin{equation}
\MaxDist(o', E) = \dist(o, e)
\end{equation}
\begin{equation}
\MinDist(o', B) = \dist(o, b)
\end{equation}

    By our assumption, $\dist(o, e) < \dist(o, b)$.
\end{proof}

\begin{lemma}
    \label{app_right}
    \begin{gather*}
    \left[\forall o' \in \Corners(O) : \MaxDist(o', E) < \MinDist(o', B) \right] \implies \\
    \forall o \in O : \forall e \in E : \forall b \in B : \dist(o, e) < \dist(o, b)
    \end{gather*}
\end{lemma}

\begin{proof}
\strut
\\
    We assume $\forall o' \in \Corners(O) : \MaxDist(o', E) < \MinDist(o', B)$ and prove $\forall o \in O : \forall e \in E : \forall b \in B : \dist(o, e) < \dist(o, b)$.
    
    Let $\Corners(O) = c_1, c_2, \dots, c_{(2^r)}$. Since the AABB represented by $O$ is convex, we can represent an arbitrary point $o \in O$ as an affine combination of the corners:

\begin{equation}
o = \sum_{i=1}^{2^r} \alpha_i c_i
\end{equation}

    where $\forall i \in 1..2^r : \alpha_i \geq 0$ and $\alpha_1 + \alpha_2 + \dots + \alpha_{(2^r)} = 1$.
    
    Since $\MaxDist$ and $\MinDist$ are non-negative, we can rewrite this for any point $p$:

    \begin{align}
    \MaxDist(p, E) & < \MinDist(p, B) \label{a} \\
    \MaxDist(p, E)^2 & < \MinDist(p, B)^2 \\
    \MaxDist(p, E)^2 - \MinDist(p, B)^2 & < 0 \label{b}
    \end{align}
    
    Parameterizing this, let $g(p) = \MaxDist(p, E)^2 - \MinDist(p, B)^2$.
    Since $g(p)$ is convex by \Cref{g_convex}, we can apply Jensen's inequality:
    
    \begin{align}
    g\left(\sum_{i=1}^{2^r} \alpha_i c_i\right) &\leq \sum_{i=1}^{2^r} \alpha_i \cdot g\left(c_i\right) \\
    g\left(o\right) &\leq \sum_{i=1}^{2^r} \alpha_i \cdot g\left(c_i\right)
    \end{align}
    
    By the prior rewrites \cref{a} to \cref{b} we know that $g(p) < 0 \iff \MaxDist(p, E) < \MinDist(p, B)$.
    Thus, we can write:
    
    \begin{align}
        \forall o' \in \Corners(O) & : \MaxDist(o', E) < \MinDist(o', B) \\
        \forall o' \in \Corners(O) & : g(o') < 0 \\
        \forall i \in 1..2^r & : g(c_i) < 0
    \end{align}
    
    Since $\alpha_i \geq 0$ and $g(c_i) < 0$, we know that $\forall i \in 1..2^r : \alpha_i \cdot g\left(c_i\right) \leq 0$.
    Since $\sum \alpha_i = 1$, at least one term is strictly negative, so the sum is strictly negative:
    
    \begin{align}
        \sum_{i=1}^{2^r} \alpha_i \cdot g\left(c_i\right) & < 0 \\
        g(o) \leq \sum_{i=1}^{2^r} \alpha_i \cdot g\left(c_i\right) & < 0 \\
        g(o) & < 0 \\
        \MaxDist(o, E) & < \MinDist(o, B) \\
        \max_{e\in E} \dist(o, e) & < \min_{b\in B} \dist(o, b) \\
        \forall e \in E, b \in B : \dist(o, e) \leq \max_{e'\in E} \dist(o, e') & < \min_{b'\in B} \dist(o, b') \leq \dist(o, b) \\
        \forall e \in E : \forall b \in B : \dist(o, e) & < \dist(o, b)
    \end{align}
\end{proof}

\begin{lemma}
    \label{g_convex}
    For $g(p) = \MaxDist(p, E)^2 - \MinDist(p, B)^2$,
    $g(p)$ is convex.
\end{lemma}
\begin{proof}
    We first introduce the squared definitions of $\MaxDist$ and $\MinDist$:

\begin{equation}
        \MaxDist(p, M)^2 = \sum^r_{d=1} \max \,\left\{    \begin{array}{l}
            \left(\check{M}_d - p_d\right)^2, \\
            \left(\hat{M}_d - p_d\right)^2
        \end{array}
        \right\}
\end{equation}

\begin{equation}
    \begin{array}{l}
        \MinDist(p, M)^2 = \\
        \sum^r_{d=1} \left\{
        \begin{array}{lr}
            \left(\check{M}_d - p_d\right)^2, & p_d < \check{M}_d \\
            \left(\hat{M}_d - p_d\right)^2, & \hat{M}_d < p_d \\
            0, & \check{M}_d \leq p_d \leq \hat{M}_d
        \end{array}
        \right\}
    \end{array}
\end{equation}
    
    It is sufficient to prove that each dimension is individually convex, as the sums of convex functions are also convex.
    Every dimension $d$ has an associated function $h_d$ of type $\mathbb{R} \rightarrow \mathbb{R}$.
    Let $e : \mathbb{R} = E_d$ and $b : \mathbb{R} = B_d$:

\begin{equation}
    g(p) = \sum^r_{d=1} h_d(p)
\end{equation}

\begin{equation}
    \begin{array}{l}
            h_d(p) = \max\left\{
  \begin{array}{l}
    (\check{E}_d - p)^2, \\
    (\hat{E}_d - p)^2
  \end{array}\right\} - \\
            \quad\quad \left\{
  \begin{array}{lr}
    (\check{B}_d - p)^2, & p < \check{B}_d \\[3pt]
    (\hat{B}_d - p)^2, & \hat{B}_d < p \\[3pt]
    0, & \check{B}_d \leq p \leq \hat{B}_d
  \end{array}
\right\}
        \end{array}
\end{equation}

    It now suffices to prove, for all bounds $[\check{E}_d, \hat{E}_d]$ and $[\check{B}_d, \hat{B}_d]$, that $h_d(p)$ is convex.
    In this definition, the $\max$ expression is symmetric around the midpoint.
    Let $\bar{E}_d=\frac{\check{E}_d+\hat{E}_d}2$.

\begin{equation}
    \max \,\left\{    \begin{array}{l}
            \left(\check{E}_d - p\right)^2, \\
            \left(\hat{E}_d - p\right)^2
        \end{array}\right\}
    =
    \left\{
            \begin{array}{lr}
                \left(\hat{E}_d - p\right)^2, & p < \bar{E}_d \\
                \left(\check{E}_d - p\right)^2, & \bar{E}_d \leq p
    
            \end{array}
        \right\}
\end{equation}

    We can now redefine the statement as a piecewise expression:

\begin{equation}
    \begin{array}{l}
            h_d(p) =\left\{
            \begin{array}{lr}
                \left(\hat{E}_d - p\right)^2, & p < \bar{E}_d \\
                \left(\check{E}_d - p\right)^2, & \bar{E}_d \leq p
    
            \end{array} \right\} - \\
            \quad \left\{
            \begin{array}{lr}
                \left(\check{B}_d - p\right)^2, & p < \check{B}_d \\
                \left(\hat{B}_d - p\right)^2, & \hat{B}_d < p \\
                0, & \check{B}_d \leq p \leq \hat{B}_d
            \end{array}
            \right\}
        \end{array}
\end{equation}

    This function has at most three piecewise break points: $p=\bar{E}_d$, $p=\check{B}_d$, and $p=\hat{B}_d$.
    To show that piecewise function $h_d(p)$ is convex, we demonstrate that each continuous span of the function has a non-negative second derivative.
    These spans join at break points, characterized by an instantaneous change in their first derivative, where we show that the right-sided derivative is not smaller than the left-sided derivative.
    Every case of $h_d(p)$ has form $(a-x)^2-(b-x)^2$, or, when $\check{B}_d\leq p\leq \hat{B}_d$, then $(a-x)^2$.
    Both of these forms, $\forall a, b \in \mathbb{R} : f(x)=(a-x)^2-(b-x)^2$ and $\forall a : f(x)=(a-x)^2$, are convex with respect to $x$, so it suffices to only check the break points connecting these spans:
        
    \textbf{Case:} $\bar{E}_d=\check{B}_d=\hat{B}_d$ \\
    First, rewrite $\check{B}_d$ and $\hat{B}_d$ as $\bar{E}_d$:

\begin{equation}
    \begin{array}{l}
            h_d(p) =\left\{
            \begin{array}{lr}
                \left(\hat{E}_d - p\right)^2, & p < \bar{E}_d \\
                \left(\check{E}_d - p\right)^2, & \bar{E}_d \leq p
    
            \end{array} \right\} - \\
            \quad \left\{
            \begin{array}{lr}
                \left(\bar{E}_d - p\right)^2, & p < \bar{E}_d \\
                \left(\bar{E}_d - p\right)^2, & \bar{E}_d < p \\
                0, & \bar{E}_d \leq p \leq \bar{E}_d
            \end{array}
            \right\}
    \end{array}
\end{equation}
    
    Second, simplify piecewise:

\begin{equation}
    h_d(p) =\left\{
            \begin{array}{lr}
                \left(\hat{E}_d - p\right)^2, & p < \bar{E}_d \\
                \left(\check{E}_d - p\right)^2, & \bar{E}_d \leq p
    
            \end{array} \right\} - \left(\bar{E}_d - p\right)^2
\end{equation}
\begin{equation}
    \frac{d}{dp} \left[(\hat{E}_d-p)^2-(\bar{E}_d -p)^2\right]=2\bar{E}_d-2\hat{E}_d
\end{equation}
\begin{equation}
    \lim_{p\to \bar{E}_d^{-}} 2\bar{E}_d-2\hat{E}_d = 2\bar{E}_d-2\hat{E}_d
\end{equation}
\begin{equation}
    \frac{d}{dp} \left[(\check{E}_d-p)^2-(\bar{E}_d -p)^2\right]=2\bar{E}_d-2\check{E}_d
\end{equation}
\begin{equation}
    \lim_{p\to \bar{E}_d^{+}} 2\bar{E}_d-2\check{E}_d = 2\bar{E}_d-2\check{E}_d
\end{equation}
    
    Finally, the left-sided derivative is less than or equal to the right-sided derivative:
\begin{equation}
    2\bar{E}_d-2\hat{E}_d \leq 2\bar{E}_d-2\check{E}_d
\end{equation}
\begin{equation}
    -\hat{E}_d \leq -\check{E}_d
\end{equation}
\begin{equation}
    \check{E}_d \leq \hat{E}_d
\end{equation}

    \begin{table}[]
        \begin{tabular}{|lccr|}
        \hline
        $x$ & $\lim_{p\to x^{-}} \frac{d}{dp}h_d(p)$ & $\lim_{p\to x^{+}} \frac{d}{dp}h_d(p)$ & asm. \\
        \hline\hline
        \multicolumn{4}{|l|}{\textbf{Case: }$\bar{E}_d<\check{B}_d<\hat{B}_d$} \\
        \hline
        $\bar{E}_d$ & $2\check{B}_d - 2\hat{E}_d$ & $2\check{B}_d - 2\check{E}_d$ & $\check{E}_d \leq \hat{E}_d$ \\
        $\check{B}_d$ & $2\check{B}_d - 2\check{E}_d$ & $2\check{B}_d-2\check{E}_d$ &  \\
        $\hat{B}_d$ & $2\hat{B}_d - 2\check{E}_d$ & $2\hat{B}_d - 2\check{E}_d$ &  \\
        \hline\hline
        \multicolumn{4}{|l|}{\textbf{Case: }$\check{B}_d<\bar{E}_d<\hat{B}_d$} \\
        \hline
        $\check{B}_d$ & $2\check{B}_d - 2\hat{E}_d$ & $2\check{B}_d-2\hat{E}_d$ &  \\
        $\bar{E}_d$ & $2\bar{E}_d - 2\hat{E}_d$ & $2\bar{E}_d - 2\check{E}_d$ & $\check{E}_d \leq \hat{E}_d$ \\
        $\hat{B}_d$ & $2\hat{B}_d - 2\check{E}_d$ & $2\hat{B}_d-2\check{E}_d$ &  \\
        \hline\hline
        \multicolumn{4}{|l|}{\textbf{Case: }$\check{B}_d<\hat{B}_d<\bar{E}_d$} \\
        \hline
        $\check{B}_d$ & $2\check{B}_d - 2\hat{E}_d$ & $2\check{B}_d-2\hat{E}_d$ &  \\
        $\hat{B}_d$ & $2\hat{B}_d - 2\hat{E}_d$ & $2\hat{B}_d - 2\hat{E}_d$ &  \\
        $\bar{E}_d$ & $2\hat{B}_d - 2\hat{E}_d$ & $2\hat{B}_d-2\check{E}_d$ & $\check{E}_d \leq \hat{E}_d$ \\
        \hline\hline
        \multicolumn{4}{|l|}{\textbf{Case: }$\bar{E}_d<\check{B}_d=\hat{B}_d$} \\
        \hline
        $\bar{E}_d$ & $2\check{B}_d - 2\hat{E}_d$ & $2\check{B}_d - 2\check{E}_d$ & $\check{E}_d \leq \hat{E}_d$ \\
        $\check{B}_d$ & $2\check{B}_d - 2\check{E}_d$ & $2\hat{B}_d - 2\check{E}_d$ &  \\
        \hline\hline
        \multicolumn{4}{|l|}{\textbf{Case: }$\bar{E}_d=\check{B}_d<\hat{B}_d$} \\
        \hline
        $\bar{E}_d$ & $2\check{B}_d-2\hat{E}_d$ & $2\bar{E}_d - 2\check{E}_d$ & $\check{E}_d \leq \hat{E}_d$ \\
        $\hat{B}_d$ & $2\hat{B}_d-2\check{E}_d$ & $2\hat{B}_d-2\check{E}_d$ &  \\
        \hline\hline
        \multicolumn{4}{|l|}{\textbf{Case: }$\check{B}_d<\bar{E}_d=\hat{B}_d$} \\
        \hline
        $\check{B}_d$ & $2\check{B}_d - 2\hat{E}_d$ & $2\check{B}_d - 2\hat{E}_d$ &  \\
        $\bar{E}_d$ & $2\bar{E}_d - 2\hat{E}_d$ & $2\hat{B}_d - 2\check{E}_d$ & $\check{E}_d \leq \hat{E}_d$ \\
        \hline\hline
        \multicolumn{4}{|l|}{\textbf{Case: }$\check{B}_d=\hat{B}_d<\bar{E}_d$} \\
        \hline
        $\check{B}_d$ & $2\check{B}_d - 2\hat{E}_d$ & $2\hat{B}_d - 2\hat{E}_d$ &  \\
        $\bar{E}_d$ & $2\hat{B}_d - 2\hat{E}_d$ & $2\hat{B}_d - 2\check{E}_d$ & $\check{E}_d \leq \hat{E}_d$ \\
        \hline
        \end{tabular}
        \vspace{1mm}
        \caption{Remaining Proof Cases for \Cref{g_convex}}
        \label{tab:cases}
    \end{table}

    $\check{E}_d \leq \hat{E}_d$ is an existing assumption, proving this case.
    The remaining cases are repetitive and uninsightful, so we outline them in \Cref{tab:cases}.
    If more than the mere rewriting of equal expressions is needed for completion, we list the assumptions used (asm.).
\end{proof}

\subsection{Optimized AllPointsCloser Algorithm}

\begin{algorithm}
    \caption{\textsc{AllPointsCloser} (optimized)}
    \label{fig:algo-v2}
    \begin{algorithmic}[1]
        \Statex \textbf{Input:} origin AABB $O$, evaluation AABB $E$, basis AABB $B$
        \Statex \textbf{Output:} \textbf{true} if all points in $O$ have nearest neighbor in $E$
        \Function{AllPointsCloser}{$O,E,B$}
            \State $s \gets 0$                         \Comment{running sum}
            \For{$i \gets 1$ to $r$}                   \Comment{loop over each dimension}
                \State $(o_{\min},o_{\max}) \gets O[i]$
                \State $(e_{\min},e_{\max}) \gets E[i]$
                \State $(b_{\min},b_{\max}) \gets B[i]$
                \State
                \If{$o_{\min} < b_{\min}$}
                    \State $d_{\min B} \gets (b_{\min}-o_{\min})^{2}$
                \ElsIf{$b_{\max} < o_{\min}$}
                    \State $d_{\min B} \gets (b_{\max}-o_{\min})^{2}$
                \Else
                    \State $d_{\min B} \gets 0$
                \EndIf
                \State
                \If{$o_{\max} < b_{\min}$}
                    \State $d_{\max B} \gets (b_{\min}-o_{\max})^{2}$
                \ElsIf{$b_{\max} < o_{\max}$}
                    \State $d_{\max B} \gets (b_{\max}-o_{\max})^{2}$
                \Else
                    \State $d_{\max B} \gets 0$
                \EndIf
                \State
                \State $d_{\min E} \gets \max\bigl((e_{\max}-o_{\min})^{2},\,(e_{\min}-o_{\min})^{2}\bigr)$
                \State $d_{\max E} \gets \max\bigl((e_{\max}-o_{\max})^{2},\,(e_{\min}-o_{\max})^{2}\bigr)$
                \State $s \gets s + \min\bigl(d_{\min B}-d_{\min E},\; d_{\max B}-d_{\max E}\bigr)$
            \EndFor
            \State \Return $(s > 0)$
        \EndFunction
    \end{algorithmic}
\end{algorithm}

We propose an optimized algorithm that runs in $O(R)$ time, shown in \Cref{fig:algo-v2}.
\Cref{fig:algo} runs in $O(R\cdot 2^R)$ time due to testing each of the $2^R$ corner points.
Our optimized code is derived by taking the original \Cref{fig:algo}, inlining the MinDist and MaxDist functions, and searching for which corner is most aligned to find an exception.
The unoptimized algorithm searches for one corner point of the origin where the MinDist to the basis is less than or equal to the MaxDist to the evaluation.
The optimized algorithm solves for each dimension independently by finding which direction, towards either the low or high bound for a given dimension, maximizes the difference between these two measures.
We then merge this search with the computation of a maximum $\MaxDist(c_x,E)-\MinDist(c_x,B)$ in the running sum.

Formally, the optimized check is:
\begin{multline}
\textsc{AllPointsCloser}(O, E, B) \iff \\
  0 < \sum_{d=1}^{r} \min\!\left\{
    \begin{array}{l}
      \gapLo{d},\\
      \gapHi{d}
    \end{array}
  \right\}
\end{multline}
where $\MinDist^2$ and $\MaxDist^2$ are applied per-dimension to the scalar bound endpoints $\check{O}_d$, $\hat{O}_d$ and the dimension intervals $B_d$, $E_d$.

The execution time of the optimized AllPointsCloser algorithm is shown in \Cref{tab:pruning_algo_times}.
The code was written in Rust and ran on an AMD Ryzen 3900X.
On the most commonly used numeric types and spatial dataset dimensionalities, our pruning method executes within 10ns (floating point and integers up to 64 bits; 2d to 4d).

\begin{table}[]
    \centering
    \begin{tabular}{l|rrrrrrr}
      \hline
    \textbf{Numeric Type} & \textbf{2d} & \textbf{3d} & \textbf{4d} & \textbf{8d} & \textbf{16d} & \textbf{24d} & \textbf{32d} \\ 
      \hline
    Int8 &   6 &   9 &  12 &  25 &  52 &  77 & 106 \\ 
      Int16 &   5 &   7 &   8 &  18 &  38 &  58 &  81 \\ 
      Int32 &   7 &   8 &  10 &  20 &  42 &  63 &  87 \\ 
      Int64 &   5 &   7 &  10 &  20 &  41 &  62 &  82 \\ 
      Int128 &  12 &  19 &  28 &  49 &  96 & 144 & 202 \\ 
      Float32 &   4 &   6 &   9 &  16 &  32 &  45 &  68 \\ 
      Float64 &   4 &   6 &   9 &  19 &  30 &  47 &  69 \\ 
      Rational64 &  62 &  97 & 126 & 244 & 491 & 710 & 955 \\ 
       \hline
    \end{tabular}
    \vspace{1mm}
    \caption{Optimized pruning test runtime in nanoseconds.} 
    \label{tab:pruning_algo_times}
\end{table}

\subsection{Generalized Application of AllPointsCloser}\label{sec:partial-order}

It is essential to consider the case when $k$ exceeds the size of a partition, when the bounds are not minimum, or when predicates are used, making the number of potential neighbors in a partition unknown.
AllPointsCloser constitutes a strict partial order relation `$<$' between bounds for each independent origin partition.
$$E < B \; |\; O  \iff 
    \forall o \in O : \forall e \in E : \forall b \in B : \dist(o, e) < \dist(o, b)
$$
We use this to denote that ``$E$ is closer than $B$ with respect to $O$''.
A proof is provided in \Cref{sec:partial-ordering}.
By ordering partitions in this way, we can construct a DAG between train partitions for each test/origin partition.
Many kNN approaches implicitly use this property with existing pruning mechanisms, such as R-Tree kNN search.
In the layout shown in \Cref{fig:direction_contradiction}, assuming all partitions have one point, kNN with $k=2$ only needs \textbf{P1} and \textbf{P2}.
Similarly, if an additional predicate is added to a scan, then partition row counts are only an upper bound, and the number of partitions needed is unbounded.
Partitions can then be searched along the DAG's topological order until $k$ is reached.

\subsubsection{Partial Ordering Proof}\label{sec:partial-ordering}

\begin{lemma}
All points proximity is irreflexive.

$$\lnot(A < A \;|\; O)$$
\end{lemma}
\begin{proof}
We rewrite the definition and prove a contradiction:
$$\forall o \in O : \forall a_1 \in A : \forall a_2 \in A : \dist(o, a_1) < \dist(o, a_2)$$

This statement is always false when $a_1=a_2$.
\end{proof}

\begin{lemma}
All points proximity is transitive.

$$
(A < B \;|\; O) \land (B < C \;|\; O) \implies (A < C \;|\; O)
$$
\end{lemma}
\begin{proof}
Assume $(A < B \;|\; O)$ and $(B < C \;|\; O)$.
\begin{align*}
&\forall o \in O : \forall a \in A : \forall b \in B : \dist(o,a) < \dist(o,b)\\
&\forall o \in O : \forall b \in B : \forall c \in C : \dist(o,b) < \dist(o,c)
\end{align*}
We want to show $\forall o \in O : \forall a \in A : \forall c \in C : \dist(o,a) < \dist(o,c)$.

Fix any $o \in O$, $a \in A$, and $c \in C$. Since $B$ is the set we compare to in the second assumption, let us choose any $b \in B$. Then by the first assumption we have
\[
\dist(o,a) < \dist(o,b),
\]
and by the second assumption we have
\[
\dist(o,b) < \dist(o,c).
\]
By the transitivity of the usual order of the real numbers,
\[
\dist(o,a) < \dist(o,c).
\]
Since $o$, $a$, and $c$ were fixed arbitrarily from a universal quantifier, this shows that $(A < C \;|\; O)$ holds, proving transitivity.
\end{proof}

\section{Conclusion}
Data warehouses are becoming the predominant architecture for querying large datasets.
Rather than maintaining persistent indexes, they use columnar partition distribution statistics to accelerate lookups and joins.
AkNN joins can benefit from partition-bounds-only pruning; however, existing pruning methods were designed for in-index pruning and are suboptimal for this use case.

We present the AllPointsCloser algorithm to determine if every point in one partition is closer to every point in another partition for all points in an origin partition.
We prove our three-bound test is optimal; if one partition is always closer, then AllPointsCloser proves it.
Additionally, AllPointsCloser is not subject to the restrictions of MinMaxDist~\cite{roussopoulos_nearest_1995} or NXNDist~\cite{4221754}, which are not applicable when $k>1$, bounds are not minimum bounds, or when using additional predicates in a query.
This makes AllPointsCloser a simple, optimal, and pragmatic improvement for spatial AkNN joins across partitioned datasets, regardless of the subsequent join algorithm.

\paragraph{Availability}
A Rust implementation and Lean 4 formal proof are available at \url{https://github.com/dominikWin/apc-aknn} and in the arXiv ancillary files.

\bibliographystyle{ACM-Reference-Format}
\bibliography{main}

\end{document}